\documentclass{article}
\usepackage{iclr2027_conference,times}

\usepackage{amsmath,amssymb,amsthm}
\usepackage{graphicx,booktabs,array,tabularx,needspace}
\usepackage{xcolor}
\usepackage{mdframed}
\usepackage{hyperref,url}
\hypersetup{
    colorlinks=false,
    pdfborder={0 0 0},
    pdftitle={What Response Marginals Miss: Adaptive Query Complexity of Functional Backdoor Recovery},
    pdfauthor={Yunjae Hwang and Byoungjin Seok}
}
\theoremstyle{definition}
\newtheorem{definition}{Definition}
\newtheorem{theorem}{Theorem}
\newtheorem{proposition}{Proposition}
\newtheorem{corollary}{Corollary}
\newtheorem{lemma}{Lemma}

\mdfsetup{
    backgroundcolor=gray!16,
    hidealllines=true,
    innerleftmargin=4pt,
    innerrightmargin=4pt,
    innertopmargin=4pt,
    innerbottommargin=4pt,
    skipabove=6pt,
    skipbelow=6pt
}

\surroundwithmdframed{definition}
\surroundwithmdframed{theorem}
\surroundwithmdframed{proposition}
\surroundwithmdframed{corollary}
\surroundwithmdframed{lemma}

\newcommand{\ind}{\mathbf 1}
\newcommand{\HH}{\mathcal H}
\newcommand{\QQ}{\mathcal Q}
\newcommand{\ZZ}{\mathcal Z}

\newcommand{\FAST}{\mathrm F}
\newcommand{\SLOW}{\mathrm S}
\newcommand{\BIN}{\mathrm B}

\title{What Response Marginals Miss: Adaptive Query Complexity of Functional Backdoor Recovery}

\author{
Yunjae Hwang \\
School of Cybersecurity \\
Korea University \\
Seoul, Republic of Korea \\
\texttt{hyj019@korea.ac.kr} \\
\And
Byoungjin Seok \\
Department of Convergence Security \\
Hansung University \\
Seoul, Republic of Korea \\
\texttt{bjseok@hansung.ac.kr}
}

\iclrfinalcopy

\begin{document}
\maketitle
\maketitle
\lhead{}

\begin{abstract}
Functional backdoor recovery finds any trigger whose attack success rate is at least a given threshold rather than to recover the planted trigger. We study the minimum number of queries required for this task under label feedback which returns the predicted class label. We construct two finite families of victim models that have exactly the same attack success rate for every victim and trigger candidate. The distribution of returned labels for every query is also identical under a uniformly chosen victim. These families form an explicit counterexample that despite the matched quantities, their optimal adaptive query complexities are \(\Theta(\log H)\) and \(\Theta(H)\) where \(H\) is the number of possible victims. The difference arises because the same non-target labels are associated with different sets of victims, so successive queries eliminate possible victims at different rates. This separation disappears when the response is reduced to binary feedback, which reports only whether the target label is returned. The separation also persists for every fixed failure probability below one. Finally, we realize the same recovery problems with trained CIFAR-10 ResNet-18 classifiers and verify the predicted optimal query budgets. These results show that attack success rate and the distribution of returned labels for each query are insufficient to determine the query complexity of functional backdoor recovery.
\end{abstract}

\section{Introduction}
\label{sec:introduction}

Trigger inversion is a common approach to backdoor inspection in which a defender searches for transformations that activate the target behavior of a suspicious classifier
\citep{wang2019neural,liu2019abs,tao2022better,wang2022rethinking,wang2023unicorn}.
Black-box inspection has also been studied under limited access to model outputs
\citep{chen2019deepinspect,dong2021black,guo2022aeva,le2024restore,popovic2025debackdoor}.
For trigger inversion, recovering the exact planted trigger is not necessary: approximate or alternative transformations can activate the same malicious behavior
\citep{qiu2024belt,abad2026removing}.
This motivates functional backdoor recovery, where success requires any admissible transformation whose attack success rate (ASR) is at least a prescribed threshold.

The difficulty of functional recovery depends not only on the backdoor itself but also on the information available to the defender.
Black-box inspection methods study this problem under limited access to model outputs
\citep{chen2019deepinspect,dong2021black,guo2022aeva,le2024restore,popovic2025debackdoor},
while recent attacks have been designed to resist trigger inversion or related inspection methods
\citep{zhu2024gradient,qiu2024belt,cheng2024lotus}.
These studies analyze particular inspection methods or attacks designed to resist inversion rather than characterize the optimal adaptive query budget under a fixed query model.
Decision-based black-box attacks show that class labels alone can guide adaptive search
\citep{brendel2018decisionbased,chen2020hopskipjumpattack},
while adaptive testing and decision-region methods study how query outcomes reduce the remaining possibilities
\citep{cicalese2014diagnosis,javdani2014decision,8104072,chen2017correlated}.
These works do not address the controlled comparison considered here. We ask whether functional recovery problems with the same complete candidate success relation and the same response marginal for every query can nevertheless have different optimal adaptive query budgets.

To isolate this question, we compare victim families that agree on ASR for every victim and candidate and have the same response marginal for every query under a uniformly drawn hidden victim.
For deterministic responses, this also matches each query's response entropy and its initial mutual information with the hidden victim.
The marginals do not determine which victims produce each response and therefore do not determine how the remaining victims are partitioned after successive observations.
We show that this missing information matters by constructing an explicit counterexample in which the matched quantities are identical but the optimal adaptive query budgets are different.

The counterexample consists of two families called FAST and SLOW.
FAST queries reveal complementary information about the hidden victim, whereas SLOW contains a hard subfamily for which a consistent response path removes at most one remaining victim per query.
Consequently, the optimal recovery budgets grow logarithmically and linearly with the family size despite the matched success relation and response marginals.
The separation persists for every fixed failure probability below one but disappears when each response is reduced to binary feedback.
We also show that the exact finite recovery budgets predicted by the construction hold for trained CIFAR-10 ResNet-18 classifiers at the tested problem sizes.

\textbf{Our contributions.}
We establish exact query complexity results for functional backdoor recovery
under label feedback:

\begin{itemize}

\item
We give an explicit counterexample showing that the optimal query budget for functional backdoor recovery is not determined by the ASR of the trigger candidates or by theresponse marginal of each query. We construct two finite victim families that match both quantities exactly and show that their optimal zero-error recovery budgets become different as \(\Theta(\log H)\) and \(\Theta(H)\).

\item
We strengthen the separation in two directions: For every fixed failure probability \(\delta<1\), two families retain \(\Theta(\log H)\) and \(\Theta(H)\) query complexity, respectively. After label feedback is reduced to binary feedback, both families have the same exact optimal budget, showing that the separation is due to the information carried by non-target class labels rather than differences in which candidates succeed.

\item
We realize the finite recovery problems with trained convolutional classifiers and verify the conditions that transfer the theoretical query bounds to the learned models. The resulting families attain the exact predicted query budgets at \(H=12,24,48\) under the fixed-source query model.

\end{itemize}

Together, these results show that the complete ASR function and the response marginal of each query under a uniform hidden victim prior are insufficient to determine the optimal adaptive query budget. Which victim models produce each non-target label can change the number of queries required for recovery.

\section{Related Work}
\label{sec:related}

\paragraph{Trigger inversion and functional recovery.}
Backdoors associate a planted trigger with a target prediction while preserving normal behavior on clean inputs \citep{gu2017badnets,chen2017targeted}. Inversion methods reconstruct effective patterns through optimization, neuron stimulation, and constraints in input or feature space \citep{wang2019neural,liu2019abs,tao2022better,wang2022rethinking,wang2023unicorn,sun2023single,xu2024towards}. Black-box methods include DeepInspect, B3D, AEVA, REStore, and DeBackdoor, but they differ substantially in the model outputs available to the defender \citep{chen2019deepinspect,dong2021black,guo2022aeva,le2024restore,popovic2025debackdoor}. B3D assumes class-probability outputs, AEVA uses only the final predicted label, REStore includes a setting with label feedback, and DeBackdoor optimizes a probability-based surrogate of attack success. These works design particular detection or reconstruction procedures under restricted model access and analyze particular recovery or detection procedures. We instead study the minimum worst case number of queries under label feedback required by any adaptive procedure to return a candidate whose ASR is at least a prescribed threshold.

\paragraph{Inversion resistance and fundamental limits.}
Recent work has shown that the effectiveness of the planted trigger and the difficulty of recovering a working trigger need not move together. GRASP reduces the effective radius of the planted trigger and makes gradient-based inversion more difficult without substantially weakening the backdoor itself, while BELT suppresses fuzzy variants of the trigger so that the planted trigger remains effective but nearby alternatives are less likely to activate the target behavior \citep{zhu2024gradient,qiu2024belt}. These results motivate asking how difficult it is to recover any trigger that still activates the backdoor rather than judging difficulty from the planted trigger alone. At the same time, alternative-trigger analysis shows that patterns different from the planted trigger can activate the same backdoor \citep{abad2026removing}. This supports treating recovery of the malicious behavior separately from exact recovery of the planted trigger. Cryptographic constructions go further by showing that backdoors can be made computationally undetectable or non-replicable for efficient adversaries \citep{goldwasser2022planting,kalavasis2024injecting}. Our question is different: we place no restriction on computation and ask how many model predictions are required to recover a trigger that meets a given success threshold.

\paragraph{Query complexity and decision regions.}
Gluch and Urbanke treat the number of black-box queries as the main resource and ask what a computationally unbounded attacker can achieve with limited model access \citep{gluch2021query}. Work on adaptive testing further shows that the hidden hypothesis need not be identified once one decision is valid for all hypotheses still consistent with the observations \citep{golovin2010near,javdani2014decision}. This directly matches functional recovery where it is enough to find one candidate that succeeds for every victim still consistent with the returned labels. Prior work also shows that greedily reducing uncertainty can be far from optimal when dependencies among tests matter \citep{chen2017correlated}. Blackwell's comparison of experiments instead considers the distribution of observations for each possible hidden state \citep{blackwell1953equivalent}, which is stronger than matching the distribution obtained after averaging over one prior. We show that even when every candidate has the same ASR for every victim and every query returns labels with the same distribution under a uniform victim, the optimal adaptive query budget can still differ.

\section{Functional Backdoor Recovery}\label{sec:problem}

\subsection{Functional Recovery and Query Budget}\label{sec:Functional Recovery and Query Budget}

Consider a classifier \(f:\mathcal X\to\mathcal Y\) and a known target label \(t\in\mathcal Y\). Each candidate \(z\in\mathcal Z\) specifies a transformation \(T_z:\mathcal X\to\mathcal X\). We assume that \(\mathcal Z\) and the transformations \(\{T_z:z\in\mathcal Z\}\) are public. Source inputs are drawn from a distribution \(\mathcal D_{\neg t}\) over inputs whose ground-truth labels are not \(t\).

\begin{definition}[\textbf{Functional recovery}]
For \(\tau\in(0,1]\), define
\begin{equation}
g_f(z)
=
\Pr_{X\sim\mathcal D_{\neg t}}
\!\left[f(T_z(X))=t\right],
\qquad
\mathcal F_\tau(f)
=
\{z\in\mathcal Z:g_f(z)\ge\tau\}.
\label{eq:functional_set}
\end{equation}
A recovery procedure succeeds if its output belongs to
\(\mathcal F_\tau(f)\).
\end{definition}
Thus recovery does not require the planted trigger; any
\(z\in\mathcal F_\tau(f)\) is a functional trigger candidate and a valid output.

To measure the query complexity of functional recovery, we consider black-box
queries under label feedback. A query model \(\mathcal I\) specifies how the source input in each query is obtained. Depending on \(\mathcal I\), the source may be fixed, sampled from \(\mathcal D_{\neg t}\), or chosen by the algorithm.  Given a source input \(x\) and a trigger
candidate \(z\), a query returns the predicted class label \(f(T_z(x))\) under label feedback. Under binary feedback, it returns \(\mathbf 1\{f(T_z(x))=t\}\). Each model evaluation
counts as one query, while computation between queries is unrestricted; the
returned trigger candidate need not have been queried.


Let \(\mathcal H\) be an index set and
\(\{f_h\}_{h\in\mathcal H}\) a known family of possible victim models.
The algorithm knows this family but not the actual index \(h\).
For each \(h\in\mathcal H\), let $F(h)=\mathcal F_\tau(f_h)$
denote the set of trigger candidates that satisfy the functional recovery
criterion for \(f_h\). We measure how many queries are sufficient to return such a trigger
candidate for every possible victim model. For an algorithm
\(\mathcal A\) operating under \(\mathcal I\), let
\(T_{\mathcal A}\) denote the number of queries made before it returns
\(Z_{\mathcal A}\). If the algorithm does not return, we set
\(T_{\mathcal A}=\infty\) and \(Z_{\mathcal A}=\bot\notin\mathcal Z\).

\begin{definition}[\textbf{Recovery budget}]
For \(\delta\in[0,1)\), define
\begin{equation}
Q_{\mathcal A}(F,\delta)
=
\inf\left\{
q\in\mathbb N_0:
\Pr_h\!\left(
T_{\mathcal A}\le q,\,
Z_{\mathcal A}\in F(h)
\right)
\ge 1-\delta
\quad
\forall h\in\mathcal H
\right\},
\label{eq:algorithm_budget}
\end{equation}
where for each fixed \(h\), the probability is over the randomness of
\(\mathcal A\) and any random source sampling specified by \(\mathcal I\).
The optimal recovery budget is
\begin{equation}
Q^\star_{\mathcal I}(F,\delta)
=
\inf_{\mathcal A}Q_{\mathcal A}(F,\delta),
\label{eq:optimal_budget}
\end{equation}
where the infimum is over all algorithms allowed by \(\mathcal I\).
\label{def:budget}
\end{definition}
Thus, \(Q_{\mathcal A}(F,\delta)\) is the smallest query budget under which
\(\mathcal A\) returns a trigger candidate satisfying the functional recovery
criterion with probability at least \(1-\delta\) for every possible victim
model. The quantity \(Q^\star_{\mathcal I}(F,\delta)\) is the minimum
such budget over all algorithms permitted by the query model.

\subsection{Independent Sampling Baseline}
\label{sec:independent_sampling}

The optimal recovery budget in Section~\ref{sec:Functional Recovery and Query Budget} allows the algorithm to choose
queries adaptively based on previously returned labels. As a simple baseline,
consider a procedure that samples trigger candidates independently
from a fixed distribution \(Q_Z\) over \(\mathcal Z\) and evaluates each
sampled candidate on fresh source inputs.

\begin{proposition}[\textbf{IID recovery baseline}]
\label{prop:iid_recovery}

Fix a victim model \(f\) and a distribution \(Q_Z\) over
\(\mathcal Z\).Let $\rho
=
\Pr_{Z\sim Q_Z}[g_f(Z)\ge\tau]
>0,
$
and suppose that for some \(\eta>0\),
$
\Pr_{Z\sim Q_Z}
\!\left(
|g_f(Z)-\tau|\ge\eta
\right)
=1.
$
Let \(Z_1,\ldots,Z_L\stackrel{\mathrm{iid}}{\sim}Q_Z\), and independently,
for each \(i\in[L]\), let
\(X_{i,1},\ldots,X_{i,r}\stackrel{\mathrm{iid}}{\sim}\mathcal D_{\neg t}\).
Define
$
\widehat g_i
=
\frac{1}{r}
\sum_{j=1}^{r}
\mathbf 1\!\left\{
f(T_{Z_i}(X_{i,j}))=t
\right\}.
$
If the procedure returns any \(Z_i\) with \(\widehat g_i\ge\tau\), then
\begin{equation}
\Pr[\mathrm{failure}]
\le
(1-\rho)^L+2Le^{-2r\eta^2}.
\label{eq:iid_failure_bound}
\end{equation}
Hence, for any \(\delta\in(0,1)\),
\begin{equation}
L=
\left\lceil
\rho^{-1}\log\frac{2}{\delta}
\right\rceil,
\qquad
r=
\left\lceil
\frac{1}{2\eta^2}
\log\frac{4L}{\delta}
\right\rceil
\label{eq:iid_budget}
\end{equation}
suffice to return a trigger candidate in \(\mathcal F_\tau(f)\) with
probability at least \(1-\delta\), using \(Lr\) model queries.
\end{proposition}

Appendix~\ref{app:iid} proves the bound. The first term in~\eqref{eq:iid_failure_bound} is the probability that none
of the \(L\) sampled trigger candidates satisfies the recovery criterion.
The second bounds errors in deciding whether the sampled candidates lie above
or below the threshold \(\tau\) from \(r\) source evaluations. Consequently,
\(\rho\) determines how many trigger candidates must be sampled, and
\(\eta\) determines how many source inputs are needed to evaluate each one
reliably.

Even though Proposition~\ref{prop:iid_recovery} gives a sufficient query budget for this independent sampling procedure, it does not characterize the optimal adaptive budget \(Q^\star_{\mathcal I}(F,\delta)\). In particular, two recovery
problems can agree on which trigger candidates satisfy the recovery criterion
while differing in how much information their returned labels provide to an
adaptive algorithm. The next section gives an explicit counterexample in which both ASR and the response marginal of each query are matched, yet the optimal adaptive query budgets differ.

\section{Exact Separation under Matched Marginals}
\label{sec:separation}

\subsection{Construction and Exact Query Budgets}
\label{sec:construction_exact_costs}

We construct two finite recovery problems, FAST and SLOW, as an explicit
counterexample showing that ASR and the marginal distribution of
each query response under a uniformly chosen victim do not determine the optimal query complexity. The two
problems match these quantities exactly but differ in which victims produce
the non-target labels.

Set $N=2^n, m=N/2, H=N+m=3\cdot2^{n-1}$ for \(n\ge2\).
There are \(N\) core victims \(c_0,\ldots,c_{N-1}\) and \(m\) auxiliary
victims \(d_0,\ldots,d_{m-1}\) with $\mathcal H
=
\{c_0,\ldots,c_{N-1},d_0,\ldots,d_{m-1}\}.$
We use a singleton source distribution, so
the source input is fixed and omitted from the notation.
Queries and output trigger candidates use the same set,
\[
\mathcal Q=\mathcal Z
=
\{p_0,\ldots,p_{m-1},v_0,\ldots,v_{N-1}\}.
\]
The candidates \(v_z\) are paired with the core victims \(c_z\), while
\(p_i\) are paired with the auxiliary victims \(d_i\).
FAST and SLOW differ only in how the non-target responses of the \(p_i\) queries divide the core victims.

We use labels \(\{0,1,2,3,4\}\), with target label \(0\). Let
\(O_h^{\mathrm F}(q)\) and \(O_h^{\mathrm S}(q)\) denote the returned labels.
For \(z\in\{0,\ldots,N-1\}\), let \(b_k(z)\) be its \(k\)-th bit, and for
\(p_i\) let \(k=i\bmod n\). On the core victims, FAST returns
\begin{equation}
O^{\mathrm F}_{c_z}(p_i)=1+b_k(z),
\label{eq:fast_core}
\end{equation}
so each \(p_i\) divides the core victims according to one bit of \(z\).
Here \(1+b_k(z)\in\{1,2\}\), so these responses are always non-target.
For SLOW, writing the core victims as
\(c_0,\ldots,c_{m-1}\) and \(c_m,\ldots,c_{2m-1}\),
\begin{equation}
O^{\mathrm S}_{c_j}(p_i)
=
\begin{cases}
1, & j=i,\\
2, & j\neq i,
\end{cases}
\qquad
O^{\mathrm S}_{c_{m+j}}(p_i)
=
\begin{cases}
2, & j=i,\\
1, & j\neq i.
\end{cases}
\label{eq:slow_core}
\end{equation}
Hence \(p_i\) isolates one victim in the first half. The second half is
arranged to preserve the same label counts as FAST; across all
\(N\) core victims, labels \(1\) and \(2\) each occur \(m\) times.
All other responses are identical:
\begin{equation}
O^{\mathrm F}_{d_j}(p_i)
=
O^{\mathrm S}_{d_j}(p_i)
=
\begin{cases}
0, & j=i,\\
3+b_k(j), & j\neq i,
\end{cases}
\qquad
O^{\mathrm F}_{h}(v_z)
=
O^{\mathrm S}_{h}(v_z)
=
\begin{cases}
0, & h=c_z,\\
1, & h\neq c_z.
\end{cases}
\label{eq:common_responses}
\end{equation}
\(p_i\) is successful only for \(d_i\), while \(v_z\) is successful
only for \(c_z\). Writing \(\phi(c_z)=v_z\) and \(\phi(d_i)=p_i\), both
families satisfy \(F(h)=\{\phi(h)\}\) with \(\phi\) a bijection. Since the remaining responses are identical in FAST and SLOW, the matched core label counts imply the same marginal distribution for every query.

Now, let \(Q_{\mathrm F}^\star(\delta)\) and
\(Q_{\mathrm S}^\star(\delta)\) denote the optimal recovery budgets for FAST
and SLOW under label feedback. Let \(Q_{\mathrm B}^\star(\delta)\) denote
the corresponding budget when each returned label \(y\) is replaced by the
target indicator \(\mathbf 1\{y=0\}\). 

\begin{theorem}[\textbf{Query complexity separation}]
\label{thm:exact_separation}
For every victim \(h\) and trigger candidate \(q\),
\[
g_h^{\mathrm F}(q)
=
g_h^{\mathrm S}(q)
=
\mathbf 1\{q=\phi(h)\}.
\]
If \(U\sim\operatorname{Unif}(\mathcal H)\), then for every
\(q\in\mathcal Q\) and \(y\in\{0,1,2,3,4\}\),
\begin{equation}
\Pr[O_U^{\mathrm F}(q)=y]
=
\Pr[O_U^{\mathrm S}(q)=y].
\label{eq:matched_marginals}
\end{equation}
Their optimal zero-error budgets under label feedback are
\begin{equation}
Q_{\mathrm F}^\star(0)=n
\quad (n\ge2),
\qquad
Q_{\mathrm S}^\star(0)=2^{n-1}-1
\quad (n\ge3).
\label{eq:exact_full_label_costs}
\end{equation}
For \(n=2\), \(Q_{\mathrm S}^\star(0)=2\). Under binary feedback,
both families satisfy
$Q_{\mathrm B}^\star(0)=H-1.$

\end{theorem}
Since \(H=3\cdot2^{n-1}\), for \(n\ge3\), $Q_{\mathrm F}^\star(0)
=
\log_2(2H/3)$
and
$Q_{\mathrm S}^\star(0)
=
H/3-1.
$
The two budgets under label feedback coincide at \(n=2,3\), and a strict separation
begins at \(n=4\). The complete proof is deferred to Appendix~\ref{app:exact_separation}.

\subsection{Limits of ASR and Response Marginals}
\label{sec:limits_matched_quantities}

The matched response marginals specify how many victims return each
label but not which victims return them. Figure~\ref{fig:response_tables}
shows this distinction at \(H=24\). In general, zero-error functional
recovery need not identify the victim if all remaining victims share a
functional trigger candidate. In the FAST and SLOW construction, however,
each victim has a different unique functional trigger candidate, so
recovery requires identifying the victim. Appendix~\ref{app:decision}
gives the general zero-error stopping condition.

\begin{figure}[t]
 \centering
 \includegraphics[width=.91\linewidth]{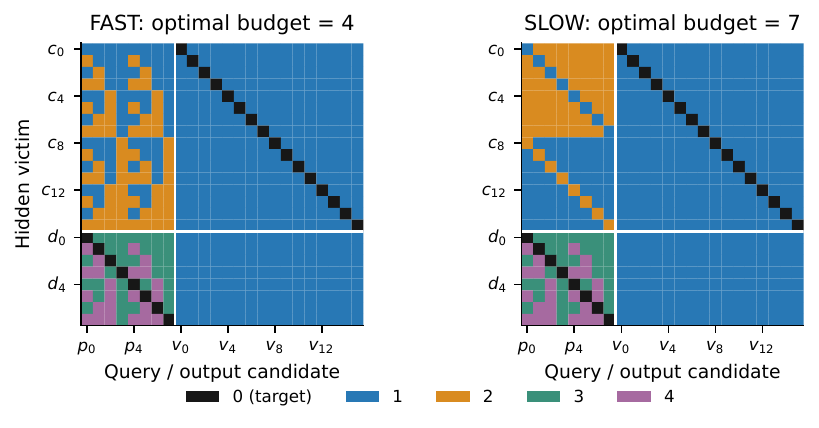}
 \caption{The \(H=24\) response tables. Rows are hidden victims and
 columns are queries, which are also output candidates. Each row and
 column has one target entry. Every column has the same label counts in
 both tables, but the association between responses and victims differs.
 The optimal zero-error recovery budgets are \(4\) and \(7\), respectively.}
 \label{fig:response_tables}
\end{figure}

In FAST, the \(p_i\) queries give complementary binary partitions of the
\(N=2^n\) core victims. Every query has at most two responses on this set,
so \(q\) queries can distinguish at most \(2^q\) cases. Since different
core victims require different functional trigger candidates,
\(2^q\ge N\) which implies \(q\ge n\). The queries
\(p_0,\ldots,p_{n-1}\) attain this bound by revealing the \(n\) bits of
the core index. In SLOW, every query is either constant or isolates at
most one victim from \(\{c_0,\ldots,c_{m-1}\}\). A worst case response
therefore eliminates at most one remaining victim per query, which
requires at least \(m-1\) queries.

Under binary feedback, the identities of the non-target labels are
discarded. A query exactly returns \(1\) when the queried trigger
candidate is the unique functional trigger candidate for the actual
victim. Every query therefore becomes an equality test and both
families require exactly \(H-1\) queries in the worst case. More
generally, label feedback cannot have a larger optimal recovery budget
than binary feedback because an algorithm may discard the class identity
of each response. Appendix~\ref{app:feedback_simulation} gives the general
simulation statement.

Because the responses are deterministic given the victim, matching the
response marginal of every fixed query also matches its response entropy
and its initial mutual information with the hidden victim. Nevertheless,
Theorem~\ref{thm:exact_separation} gives logarithmic and linear optimal
zero-error recovery budgets. Figure~\ref{fig:exact_budgets} shows how the corresponding zero-error
budgets scale with the number of possible victims.
We also give the formal
comparison under the uniform distribution over victims in Appendix~\ref{app:marginal}.

\begin{corollary}
\label{cor:Unbounded Query-Complexity Gap}
The FAST and SLOW constructions have the same ASR for every
victim--trigger candidate pair and the same response marginal for every
query, while
\[
\frac{Q_{\SLOW}^\star(0)}
     {Q_{\FAST}^\star(0)}
=
\frac{H/3-1}{\log_2(2H/3)}
\longrightarrow\infty.
\]
Consequently, these quantities do not determine the optimal zero-error
recovery budget and cannot approximate it within any fixed multiplicative
factor.
\end{corollary}
The proof can be found in Appendix~\ref{app:proof_of_corollary_1}.

\begin{figure}[t]
    \centering
    \includegraphics[width=.72\linewidth]{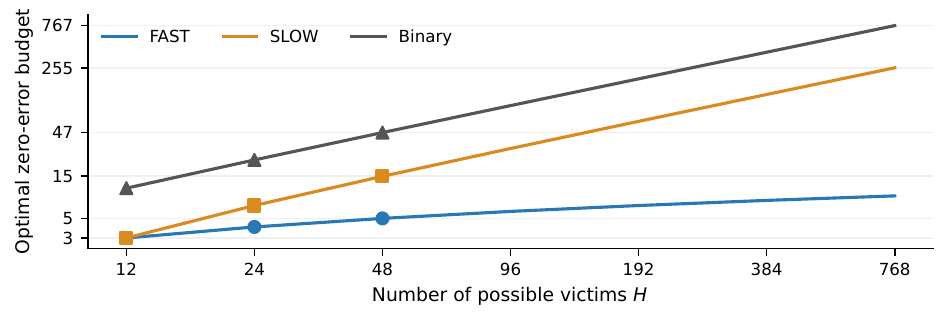}
    \caption{Exact theoretical zero-error recovery budgets as the number
    of possible victims increases. Filled markers indicate the three
    family sizes realized with trained models in Section~\ref{sec:neural};
    values at larger sizes are theoretical only. Under label feedback,
    the FAST and SLOW budgets coincide at \(H=12\) and separate from
    \(H=24\) onward. Under binary feedback, the exact budget is \(H-1\).}
    \label{fig:exact_budgets}
\end{figure}

Note that in the construction above, each
victim has a different unique functional trigger candidate, so recovery
requires identifying the victim. This identification property is not
necessary for the separation.

\begin{proposition}[\textbf{Separation without identification}]
\label{prop:overlapping_success}
For \(n\ge3\), the common responses in
Equation~\eqref{eq:common_responses} can be modified so that, for every
\(0<\tau\le1\), FAST and SLOW have the same functional success sets
\[
\widetilde F(c_z)=\{v_z,v_{(z-1)\bmod N}\},
\qquad
\widetilde F(d_j)=\{p_j,p_{(j-1)\bmod m}\}.
\]
They also have the same ASR for every victim and trigger candidate and
the same response marginal for every query under a uniformly chosen
hidden victim. Writing
\(\widetilde Q_{\mathrm F}^{\star}\),
\(\widetilde Q_{\mathrm S}^{\star}\), and
\(\widetilde Q_{\mathrm B}^{\star}\) for the corresponding optimal
budgets, their exact zero-error budgets are
\[
\widetilde Q_{\mathrm F}^{\star}(0)=n-1,
\qquad
\widetilde Q_{\mathrm S}^{\star}(0)=\frac{m}{2},
\qquad
\widetilde Q_{\mathrm B}^{\star}(0)=\frac{H}{2}-1.
\]
\end{proposition}

In this variant, an optimal FAST strategy can terminate with two possible
core victims that share a functional trigger candidate. Thus the
separation persists even when recovery does not require victim
identification. Appendix~\ref{app:overlapping_success} gives the
construction and proof.

\subsection{Fixed Failure Probability}
\label{sec:nonzero_delta}

The separation in Theorem~\ref{thm:exact_separation} is not restricted to zero-error recovery. For any fixed
failure probability \(\delta<1\), FAST retains logarithmic query
complexity while SLOW retains linear query complexity.

\begin{corollary}
\label{cor:nonzero_delta}
For \(0\le\delta<1\),
\[
\max\!\left\{
0,\,
\left\lceil
\log_2\bigl((1-\delta)N\bigr)
\right\rceil
\right\}
\le
Q_{\FAST}^\star(\delta)
\le n,
\]
and, for \(n\ge3\),
\[
\max\!\left\{
0,\,
\left\lceil
(1-\delta)m
\right\rceil-1
\right\}
\le
Q_{\SLOW}^\star(\delta)
\le m-1.
\]
Under binary feedback,
\[
Q_{\BIN}^\star(\delta)
=
\max\!\left\{
0,\,
\left\lceil
(1-\delta)H
\right\rceil-1
\right\}.
\]
Consequently, for every fixed \(\delta<1\),
\[
Q_{\FAST}^\star(\delta)=\Theta(\log H),
\qquad
Q_{\SLOW}^\star(\delta)=\Theta(H).
\]
\end{corollary}

The lower bounds follow from how many victims can be recovered correctly
within \(q\) queries. After fixing the randomness of the algorithm, a FAST
decision tree has at most \(2^q\) leaves on the \(N\) core victims, so it
can be correct for at most \(2^q\) of them. In SLOW, each query can isolate
at most one victim from the hard subset of size \(m\) and therefore at most \(q+1\) of these victims can be recovered correctly. Since Definition~\ref{def:budget}
requires success probability at least \(1-\delta\) for every victim, these
counts must be at least \((1-\delta)N\) and \((1-\delta)m\), respectively,
which gives the stated lower bounds. The zero-error strategies from
Theorem~\ref{thm:exact_separation} give the upper bounds.
Appendix~\ref{app:nonzero_delta} gives the full proof.

\section{Controlled Experiments with Trained Classifiers}
\label{sec:neural}

\subsection{Training the FAST and SLOW Classifiers}

We test whether our recovery problems of
Section~\ref{sec:construction_exact_costs} can be realized by trained classifiers on CIFAR-10. We use
ResNet-18, replacing the initial ImageNet
convolution by a \(3\times3\) stride-one convolution and removing the
initial max-pooling layer. Class \(0\) is the target class.
Each trigger candidate overwrites the bottom-right \(4\times4\) region
of an image with a binary codeword of Hamming weight eight, replicated
across the three color channels. For \(H=12,24,48\), we use public
codebooks containing \(H\) such candidates. Their minimum pairwise
Hamming distances are \(8\), \(6\), and \(6\), respectively.

For each victim \(h\), we train one FAST classifier and one SLOW
classifier from the same clean model. Clean examples retain their
original CIFAR-10 labels. For a triggered example using candidate
\(q\), the FAST and SLOW classifiers are trained toward the prescribed
labels \(O_h^{\mathrm F}(q)\) and \(O_h^{\mathrm S}(q)\), respectively.
Thus the paired classifiers share the underlying clean task and model
initialization, while the labels assigned to triggered examples follow
the corresponding response tables.
The scaling experiments use \(750\), \(1{,}500\), and \(3{,}000\)
fine-tuning steps for \(H=12,24,48\). This keeps the expected number of
triggered training examples per candidate fixed at \(16{,}000\) for
all three sizes. Appendix~\ref{app:neural_protocol} gives the complete training and evaluation
protocol.

\subsection{Exact Recovery Budgets in Trained Classifiers}
\label{sec:neural_transfer}

We now verify whether the trained classifiers realize the finite
recovery problems of Section~\ref{sec:construction_exact_costs}. Let
\(a\in\{\mathrm F,\mathrm S\}\) denote the FAST or SLOW family and let
\(f_h^a\) be the trained classifier for hidden victim \(h\).
We use a fixed source image \(x_0\), selected using only the clean base
model as specified in Appendix~\ref{app:neural_protocol}. A query
\(q\in\mathcal Q\) returns
\[
\widehat O_h^a(q)
=
f_h^a(T_q(x_0)).
\]
Each query costs one model prediction and the allowed output candidates
are the same \(H\) trigger candidates as in the finite construction.

Functionality is evaluated independently of \(x_0\).
Let
\[
\mathcal X_{\mathrm{test}}
=
\{x:\text{the ground-truth CIFAR-10 label of }x\neq0\}.
\]
This set contains \(9{,}000\) test images. For each trained classifier,
define
\begin{equation}
\widehat g_h^a(z)
=
\frac{1}{|\mathcal X_{\mathrm{test}}|}
\sum_{x\in\mathcal X_{\mathrm{test}}}
\ind\{f_h^a(T_z(x))=0\},
\qquad
\widehat F_\tau^a(h)
=
\{z:\widehat g_h^a(z)\ge\tau\}.
\label{eq:empirical}
\end{equation}
We use \(\tau=0.9\).

The learned and finite recovery problems already have the same allowed
queries, output candidates, and query cost. It therefore remains to
check two quantities: the response returned by every query and the set
of successful outputs for every hidden victim. If both agree, any
adaptive recovery algorithm sees the same recovery problem. The
following proposition formalizes this correspondence.

\begin{proposition}[\textbf{Exact recovery budget transfer}]
\label{prop:transfer}
Fix \(a\in\{\mathrm F,\mathrm S\}\). Suppose
\begin{equation}
\widehat O_h^a(q)=O_h^a(q)
\qquad
\text{for every } h\in\mathcal H
\text{ and } q\in\mathcal Q,
\label{eq:transfer_responses}
\end{equation}
and
\begin{equation}
\widehat F_\tau^a(h)=F(h)
\qquad
\text{for every } h\in\mathcal H.
\label{eq:transfer_functional}
\end{equation}
Then the learned recovery problem and the corresponding finite recovery
problem have the same optimal recovery budget for every
\(\delta\in[0,1)\) under both label feedback and binary feedback.
\end{proposition}

Equation~\eqref{eq:transfer_responses} makes every permitted query
return the same response in the two problems.
Equation~\eqref{eq:transfer_functional} makes every permitted output
succeed for the same hidden victims. Hence any adaptive recovery
algorithm has the same transcript, stopping time, and success event in
both problems. Appendix~\ref{app:neural-transfer} gives the formal
proof.

We audit exactly these two conditions. For
\(H=12,24,48\), every response at \(x_0\) agrees with its prescribed
FAST or SLOW response. Across both families, this gives
$
2(12^2+24^2+48^2)=6{,}048
$
matching responses and no mismatch. The functional set audit also
gives
$
\widehat F_{0.9}^a(h)=\{\phi(h)\}
$
for every trained classifier, which implies that all six learned families satisfy
Proposition~\ref{prop:transfer}.

\begin{table}[t]
\centering
\caption{Exact realization of the FAST and SLOW recovery problems by the trained
classifiers. For \(H=12,24,48\), all audited query responses match the
prescribed response tables and every victim satisfies
\(\widehat F_{0.9}^a(h)=\{\phi(h)\}\). The resulting zero-error recovery
budgets therefore equal the exact theoretical budgets. The reported minimum ASR is the smallest
\(\widehat g_h^a(\phi(h))\) over both families.}
\label{tab:scaling}
\begin{tabular}{c cc c ccc}
\toprule
& \multicolumn{2}{c}{Response errors}
& & \multicolumn{3}{c}{Exact recovery budget} \\
\cmidrule(lr){2-3}
\cmidrule(lr){5-7}
\(H\)
& FAST
& SLOW
& Min. functional ASR
& \(Q_{\mathrm F}^\star(0)\)
& \(Q_{\mathrm S}^\star(0)\)
& \(Q_{\mathrm B}^\star(0)\) \\
\midrule
12 & \(0/144\)  & \(0/144\)  & 0.9997 & 3 & 3  & 11 \\
24 & \(0/576\)  & \(0/576\)  & 0.9999 & 4 & 7  & 23 \\
48 & \(0/2304\) & \(0/2304\) & 0.9998 & 5 & 15 & 47 \\
\bottomrule
\end{tabular}
\end{table}

Since Proposition~\ref{prop:transfer} applies at all three sizes,
Table~\ref{tab:scaling} reports the resulting exact optimal recovery
budgets. At \(H=12\), the FAST and SLOW budgets under label feedback
equal \(3\). The separation appears at \(H=24\) where the budgets
are \(4\) and \(7\), and grows to \(5\) and \(15\) at \(H=48\). Under
binary feedback, the exact budgets are \(11\), \(23\), and \(47\). Note that the learned classifiers also preserve the clean task; the lowest clean test accuracy among all trained victim models is \(94.52\%\), while the clean model used to initialize them has a test
accuracy of \(95.09\%\).

At all three \(H=12,24,48\) cases, the trained classifiers exactly realize the
corresponding finite recovery problems, so the optimal budgets derived
in Section~\ref{sec:construction_exact_costs} transfer without approximation. These experimental realizations
are finite instances of the FAST and SLOW separation that
Theorem~\ref{thm:exact_separation} proves to grow as
\(\Theta(\log H)\) and \(\Theta(H)\), respectively.

\section{Conclusion}
\label{sec:conclusion}

Functional backdoor recovery is not determined by the ASR of trigger
candidates or the marginal distribution of returned labels for each
query. We construct FAST and SLOW families that have the same ASR for
every victim and trigger candidate and the same response marginal for
every query under a uniformly chosen hidden victim, yet their optimal
recovery budgets under label feedback are
\(\Theta(\log H)\) and \(\Theta(H)\). This separation persists for
every fixed failure probability below one and disappears under binary
feedback. The difference is not how often each label is returned, but
which victims return each label and how successive observations
partition the remaining victims.
Trained CIFAR-10 ResNet-18 classifiers at \(H=12,24,48\) exactly
realize the corresponding finite recovery problems, so the theoretical optimal
budgets apply without approximation. These results
show that ASR and response marginals alone do not characterize
functional backdoor recoverability. The association between returned
labels and possible victims across successive queries also determines
how many queries are required for recovery.

\label{end:maintext}

\section*{Reproducibility Statement}
Complete proofs of the theoretical results are provided in
Appendices~\ref{complete_proofs} and~\ref{app:feedback_simulation}. Appendix~\ref{app:neural_protocol} describes the CIFAR-10 ResNet-18
training and evaluation protocol, and Appendix~\ref{app:natural-policy} describes the
additional search-policy experiment. Source code is available at \href{https://anonymous.4open.science/r/Functional-recovery-query-complexity-F8D6}{\textcolor{blue}{here}}.


\section*{AI Use Statement}
AI tools assisted with literature search, mathematical development and checking, analysis code and manuscript editing. The core ideas, theoretical claims, and experimental design were established by the authors. All results from AI were reviewed and verified by the author.

\begingroup
\raggedright
\bibliography{iclr2027_conference}
\endgroup
\bibliographystyle{iclr2027_conference}

\clearpage

\appendix
\section{Proofs of Results}\label{complete_proofs}
\subsection{Proof of Proposition~\ref{prop:iid_recovery}}
\label{app:iid}

\begin{proof}

Let \(E_0\) be the event that none of the sampled trigger candidates belongs
to \(\mathcal F_\tau(f)\), and let
\(E_i=\{|\widehat g_i-g_f(Z_i)|\ge\eta\}\).
Since the \(Z_i\) are independent, \(\Pr(E_0)=(1-\rho)^L\).
Conditional on \(Z_i\), Hoeffding's inequality gives
\(\Pr(E_i\mid Z_i)\le 2e^{-2r\eta^2}\), and hence
\(\Pr(E_i)\le 2e^{-2r\eta^2}\).
By the margin assumption, whenever \(E_i\) does not occur,
\(\widehat g_i\ge\tau\) if and only if
\(Z_i\in\mathcal F_\tau(f)\). Therefore, failure can occur only if
\(E_0\) or at least one \(E_i\) occurs. By the union bound,
\begin{equation}
\Pr[\mathrm{failure}]
\le
(1-\rho)^L+2Le^{-2r\eta^2}.
\end{equation}

For
\[
L=
\left\lceil
\rho^{-1}\log\frac{2}{\delta}
\right\rceil,
\qquad
r=
\left\lceil
\frac{1}{2\eta^2}
\log\frac{4L}{\delta}
\right\rceil,
\]
we have \((1-\rho)^L\le e^{-\rho L}\le\delta/2\) and
\(2Le^{-2r\eta^2}\le\delta/2\), which implies \(\Pr[\mathrm{failure}]\le\delta\).
Since each of the \(L\) trigger candidates is evaluated on \(r\) source
inputs, the procedure uses \(Lr\) model queries.
\end{proof}

The bound is stated for a fixed victim model \(f\). It also holds
uniformly over a family of victim models if the same distribution
\(Q_Z\) is used and common positive lower bounds on \(\rho\) and \(\eta\)
hold for every classifier in the family. Without a positive lower bound on
\(\eta\), this argument does not yield a finite uniform value of \(r\).

\subsection{Finite Functional Recovery as a Decision Problem}
\label{app:decision}

Fix finite sets $\HH$, $\QQ$, and $\ZZ$, a deterministic response table $O$, and nonempty success sets $F(h)\subseteq\ZZ$. The family, table, and success relation specify the problem and are public; the actual index is hidden. For a nonempty version space $S\subseteq\HH$, let $D_F(S)$ be the minimum worst-case number of additional queries needed for zero-error recovery. Write
\[
 S_{q,y}=\{h\in S:O_h(q)=y\},\qquad
 \Pi_q(S)=\{S_{q,y}:S_{q,y}\ne\varnothing\}.
\]
\begin{lemma}
\label{lem:recursion}
The exact recursion is
\begin{equation}
 D_F(S)=
 \begin{cases}
 0,&\bigcap_{h\in S}F(h)\ne\varnothing,\\[2pt]
 1+\displaystyle\min_{q:\,|\Pi_q(S)|\ge2}\ \max_{C\in\Pi_q(S)}D_F(C),&\text{otherwise},
 \end{cases}
 \label{eq:bellman}
\end{equation}
where a minimum over an empty set is $+\infty$. Only queries that split $S$ are included: a constant query consumes one unit without changing the state. This restriction makes every recursive child a proper subset.
\end{lemma}

\begin{proof}
First consider deterministic strategies. A strategy using no further queries returns a fixed $z$, and is correct on $S$ exactly when $z\in\bigcap_{h\in S}F(h)$. Suppose this intersection is empty. Queries constant on $S$ can be deleted without changing any subsequent response or output. If a correct strategy with a finite worst case query budget exists, its first remaining query $q$ therefore satisfies $|\Pi_q(S)|\ge2$. For each $C\in\Pi_q(S)$ its continuation is a correct strategy on $C$, so its worst case query count $d$ satisfies
\[
 d\ge 1+\max_{C\in\Pi_q(S)}D_F(C)
 \ge 1+\min_{q':\,|\Pi_{q'}(S)|\ge2}\max_{C\in\Pi_{q'}(S)}D_F(C).
\]
Conversely, if the last minimum is finite, finiteness of $\QQ$ gives a minimizing query. Attach an optimal strategy to each of its children. Induction on $|S|$ applies because every child is a proper subset, and gives the reverse inequality. If the minimum is infinite, the preceding lower bound excludes every strategy with a finite worst case query budget. If there is no splitting query, every transcript is the same for all $h\in S$, and no output is correct on $S$. The lemma below extends the deterministic conclusion to randomized zero-error strategies.
\end{proof}

This recursion is a standard decision-tree argument. The sets $R_z=\{h:z\in F(h)\}$ are decision regions: termination requires $S\subseteq R_z$ for some $z$ \citep{javdani2014decision}. Identification is sufficient but generally unnecessary. In the FAST and SLOW construction, the successful outputs are distinct singletons, so it is also necessary. The FAST and SLOW construction uses this identification special case to separate the matched statistics from the optimal recovery budget.

\begin{lemma}\label{lemma_2}
For finite $\HH$ and a deterministic oracle, randomization cannot improve the zero-error worst case recovery budget.
\end{lemma}
\begin{proof}
Let $R$ include all random choices of a strategy that satisfies the zero-error deadline $q$. Since the oracle is deterministic, fixing $R=r$ fixes the entire strategy. For each $h$, define
\[
 E_h=\{r:T_{\mathcal A}(h,r)\le q,\ Z_{\mathcal A}(h,r)\in F(h)\}.
\]
The guarantee gives $\Pr_R(E_h)=1$. Finiteness of $\HH$ implies
\[
 \Pr_R\!\left(\bigcap_{h\in\HH}E_h\right)
 =1-\Pr_R\!\left(\bigcup_{h\in\HH}E_h^c\right)
 \ge1-\sum_{h\in\HH}\Pr_R(E_h^c)=1.
\]
Fix any $r$ in this intersection. The resulting deterministic strategy is correct within \(q\) queries for every victim. Conversely, every deterministic strategy is admissible as a randomized one. Their optimal zero-error budgets therefore agree.
\end{proof}

\subsection{Proof of Theorem~\ref{thm:exact_separation}}
\label{app:exact_separation}

\begin{proof}
Let $C=\{c_0,\ldots,c_{N-1}\}.$
By the result of Appendix~\ref{app:decision}, it suffices to
consider deterministic decision trees.

\paragraph{Success sets and response marginals.}
Equations~\ref{eq:fast_core}--\ref{eq:common_responses} give, for every
\(h\in\mathcal H\) and \(q\in\mathcal Q\),
\[
O_h^{\mathrm F}(q)=0
\quad\Longleftrightarrow\quad
O_h^{\mathrm S}(q)=0
\quad\Longleftrightarrow\quad
q=\phi(h).
\]
Since the source distribution is a singleton and \(0\) is the target label,
\[
g_h^{\mathrm F}(q)
=
g_h^{\mathrm S}(q)
=
\mathbf 1\{q=\phi(h)\}.
\]
Because \(\tau\in(0,1]\), both families therefore satisfy
\[
F(h)=\{\phi(h)\}.
\]

For each \(p_i\), bit balance in FAST gives
\[
\bigl|\{h\in C:O_h^{\mathrm F}(p_i)=1\}\bigr|
=
\bigl|\{h\in C:O_h^{\mathrm F}(p_i)=2\}\bigr|
=
m.
\]
Equation~\ref{eq:slow_core} gives the same two counts in SLOW. All responses
on auxiliary victims and all \(v_z\) responses are identical in the two
families. Therefore, for \(U\sim\operatorname{Unif}(\mathcal H)\), every
\(q\in\mathcal Q\) and every \(y\in\{0,1,2,3,4\}\),
\[
\Pr[O_U^{\mathrm F}(q)=y]
=
\Pr[O_U^{\mathrm S}(q)=y].
\]

\paragraph{FAST.}
For the upper bound, \(m=2^{n-1}\ge n\) for \(n\ge2\), so
\(p_0,\ldots,p_{n-1}\) are available. If the victim is \(c_z\), their
responses determine
\[
z
=
\sum_{i=0}^{n-1}
2^i\bigl(O_{c_z}^{\mathrm F}(p_i)-1\bigr).
\]
If the victim is auxiliary, the response to \(p_0\) belongs to
\(\{0,3,4\}\), whereas a core response belongs to \(\{1,2\}\), so the two
cases are distinguished immediately. A target response to \(p_i\) identifies
\(d_i\). If no target response occurs among
\(p_0,\ldots,p_{n-2}\), then for \(h=d_j\),
\[
j
=
\sum_{i=0}^{n-2}
2^i\bigl(O_{d_j}^{\mathrm F}(p_i)-3\bigr),
\]
because \(j<m=2^{n-1}\). Thus the victim, and hence its unique successful
trigger candidate, can be determined within \(n\) queries. Therefore
\[
Q_{\mathrm F}^\star(0)\le n.
\]

For the lower bound, restrict any correct decision tree of depth \(q\) to the
\(N=2^n\) core victims. Every permitted query has at most two possible
responses on \(C\), so the restricted tree has at most \(2^q\) nonempty
leaves. Since different core victims require different successful trigger
candidates, each leaf can contain at most one core victim. Hence
\[
N\le2^q,
\qquad
q\ge\log_2N=n.
\]
Therefore
\[
Q_{\mathrm F}^\star(0)=n.
\]

\paragraph{SLOW.}
For the lower bound, let
\[
a_i=c_i,\qquad b_i=c_{m+i},
\qquad
A=\{a_0,\ldots,a_{m-1}\}.
\]
On \(A\), every \(p_i\) is an equality test for \(a_i\), while every \(v_z\)
is either an equality test or constant. Thus, for any query, there is a
response that eliminates at most one possible victim from \(A\). Following
such a response after each query leaves at least \(m-k\) possible victims
after \(k\) queries. If \(q<m-1\), at least two victims remain consistent
with the transcript, and they require different successful trigger
candidates. Hence
\[
Q_{\mathrm S}^\star(0)\ge m-1.
\]

For the upper bound, assume \(n\ge3\) and first query \(p_0\). Response \(0\)
identifies \(d_0\). A response in \(\{3,4\}\) shows that the victim is
auxiliary. The responses to \(p_0,\ldots,p_{n-2}\), together with a possible
target response, then determine its index using the same bit representation
as in the FAST case, requiring at most \(n-1\) queries in total.

It remains to consider a core response to \(p_0\). Query \(p_1\) next. The
four possible response pairs leave
\[
\begin{array}{c|c}
(O_h^{\mathrm S}(p_0),O_h^{\mathrm S}(p_1))
&
\text{remaining core victims}
\\ \hline
(1,2) & \{a_0,b_1\}\\
(2,1) & \{b_0,a_1\}\\
(1,1) & \{b_j:j\notin\{0,1\}\}\\
(2,2) & \{a_j:j\notin\{0,1\}\}.
\end{array}
\]
Any set of \(r\) remaining core victims can be resolved using at most
\(r-1\) of their \(v_z\) candidates: a target response identifies the victim,
and after \(r-1\) non-target responses the final unqueried victim is
determined. The two-victim branches therefore require at most three queries
in total. Each of the other two branches contains \(m-2\) victims and
requires at most $2+(m-3)=m-1$
queries in total. Consequently,
\[
Q_{\mathrm S}^\star(0)
\le
\max\{n-1,3,m-1\}
=
m-1
\qquad(n\ge3).
\]
Combining the two bounds gives
\[
Q_{\mathrm S}^\star(0)=m-1=2^{n-1}-1
\qquad(n\ge3).
\]

When \(n=2\), the first query \(p_0\) either identifies an auxiliary victim
or leaves two possible core victims, which one \(v_z\) query resolves.
Conversely, every query has at most two possible responses on the four core
victims, so a depth-one decision tree has at most two nonempty leaves and
cannot distinguish all four victims. Hence $Q_{\mathrm S}^\star(0)=2.$

\paragraph{Binary feedback.}
Replacing each returned label \(y\) by \(\mathbf 1\{y=0\}\) gives, for
every \(h\in\mathcal H\) and \(q\in\mathcal Q\),
\[
\mathbf 1\{O_h^{\mathrm F}(q)=0\}
=
\mathbf 1\{O_h^{\mathrm S}(q)=0\}
=
\mathbf 1\{q=\phi(h)\}.
\]
Thus each query is an equality test for one possible victim. Along a sequence
of \(k\) non-target responses, at least \(H-k\) victims remain possible.
Since distinct victims require distinct successful trigger candidates,
\(H-1\) queries are necessary in the worst case.

For the matching upper bound, query \(H-1\) distinct trigger candidates,
returning a queried candidate immediately upon a target response and
returning the final unqueried candidate after \(H-1\) non-target responses.
Therefore
\[
Q_{\mathrm B}^\star(0)=H-1.
\]
\end{proof}

\subsubsection{Marginal Matching under the Uniform Prior}
\label{app:marginal}

For a fixed query \(q\) and
\(U\sim\operatorname{Unif}(\mathcal H)\), let
\[
Y_{\mathrm F}=O_U^{\mathrm F}(q),
\qquad
Y_{\mathrm S}=O_U^{\mathrm S}(q).
\]
Since both responses are deterministic given \(U\),
\[
\mathsf H(Y_{\mathrm F}\mid U)
=
\mathsf H(Y_{\mathrm S}\mid U)
=
0.
\]
Hence
\[
I(U;Y_{\mathrm F})=\mathsf H(Y_{\mathrm F}),
\qquad
I(U;Y_{\mathrm S})=\mathsf H(Y_{\mathrm S}).
\]
Equation~\ref{eq:matched_marginals} therefore implies
\[
I(U;Y_{\mathrm F})=I(U;Y_{\mathrm S})
\]
for every fixed query \(q\).
This equality does not determine which victims produce each response. For
example, when \(n=3\), the core victims returning label \(1\) to \(p_0\) are
$\{c_0,c_2,c_4,c_6\}$
in FAST whereas they are $\{c_0,c_5,c_6,c_7\}$
in SLOW. The two sets have the same size, but their intersections with the
response classes of later queries differ. Thus identical single-query
marginals do not imply the same partition of the remaining victims after a
sequence of observations.

If the response marginals were instead required to agree under every prior on
\(\mathcal H\), then choosing the point-mass prior at any
\(h\in\mathcal H\) would give
\[
\mathbf 1\{O_h^{\mathrm F}(q)=y\}
=
\mathbf 1\{O_h^{\mathrm S}(q)=y\}
\]
for every \(h\), \(q\), and \(y\). The two response tables would therefore
coincide entrywise. Theorem~\ref{thm:exact_separation} requires marginal
matching only under the specified uniform prior. This is weaker
than matching the state-conditioned observation laws used in comparisons
of statistical experiments \citep{blackwell1953equivalent}.

\subsection{Proof of Corollary~\ref{cor:Unbounded Query-Complexity Gap}}
\label{app:proof_of_corollary_1}
\begin{proof}
Suppose that a function of the matched ASRs and response marginals
approximates the optimal zero-error recovery budget within a fixed factor
\(c\ge1\). For each \(n\), FAST and SLOW give the same inputs to this
function, so it must assign the same value \(v\) to both families. Therefore
\[
\frac{1}{c}Q_{\FAST}^\star(0)
\le v\le
cQ_{\FAST}^\star(0),
\qquad
\frac{1}{c}Q_{\SLOW}^\star(0)
\le v\le
cQ_{\SLOW}^\star(0).
\]
These inequalities require
\[
\frac{Q_{\SLOW}^\star(0)}
     {Q_{\FAST}^\star(0)}
\le c^2.
\]
For \(n\ge3\), Theorem~\ref{thm:exact_separation} gives
\[
\frac{Q_{\SLOW}^\star(0)}
     {Q_{\FAST}^\star(0)}
=
\frac{H/3-1}{\log_2(2H/3)}
\longrightarrow\infty.
\]
Hence the inequality fails for sufficiently large \(n\), a contradiction.
\end{proof}

\subsection{Proof of Proposition~\ref{prop:overlapping_success}}
\label{app:overlapping_success}

\begin{proof}
Keep the core responses to the \(p_i\) queries unchanged and modify
the common responses as follows:
\begin{align*}
\widetilde O^a_{d_j}(p_i)
&=
\begin{cases}
0, & j=i \text{ or } j=(i+1)\bmod m,\\
3+b_{i\bmod n}(j), & \text{otherwise},
\end{cases}\\
\widetilde O^a_h(v_z)
&=
\begin{cases}
0, & h=c_z \text{ or } h=c_{(z+1)\bmod N},\\
1, & \text{otherwise},
\end{cases}
\qquad a\in\{\mathrm F,\mathrm S\}.
\end{align*}
Since the source distribution is a singleton and \(0\) is the target
label, a candidate has ASR one exactly on the victims for which it
returns \(0\). Hence, for every \(0<\tau\le1\),
\[
\widetilde F(c_z)=\{v_z,v_{(z-1)\bmod N}\},
\qquad
\widetilde F(d_j)=\{p_j,p_{(j-1)\bmod m}\}.
\]
The target responses are identical in FAST and SLOW. For each \(p_i\),
the unchanged core responses contain \(m\) occurrences of each of
labels \(1\) and \(2\) in both families, while all auxiliary responses
are identical. The responses to every \(v_z\) are also identical.
Thus the two families have the same complete ASR function and the same
response marginal for every query under a uniformly chosen victim.

\paragraph{FAST.}
Query \(p_1,\ldots,p_{n-1}\). A target response identifies a successful
queried candidate. Otherwise, the responses distinguish core from
auxiliary victims and determine the corresponding index up to its
least significant bit. The two remaining core victims share one
\(v_z\), and the two remaining auxiliary victims share one \(p_i\).
Thus
\[
\widetilde Q_{\mathrm F}^{\star}(0)\le n-1.
\]
Conversely, on the \(N\) core victims every query has at most two
responses. A depth-\(q\) decision tree therefore has at most \(2^q\)
leaves, and the output at each leaf can be successful for at most two
core victims. Zero-error recovery requires $N\le 2^{q+1},$
so \(q\ge n-1\). Lemma~\ref{lemma_2} shows that randomization cannot
improve the zero-error worst-case budget. Therefore
\[
\widetilde Q_{\mathrm F}^{\star}(0)=n-1.
\]

\paragraph{SLOW.}
For the upper bound, first query \(p_0\). A target response terminates
the search. A core response leaves one of the two sets
\[
\{c_1,\ldots,c_m\},
\qquad
\{c_{m+1},\ldots,c_{N-1},c_0\},
\]
each containing \(m\) consecutive core victims. Each set can be
partitioned into \(m/2\) adjacent pairs, and every such pair shares a
candidate \(v_z\). Querying the candidates for all but the final pair
therefore uses \(m/2\) queries in total. A non-target auxiliary response
leaves \(m/2-1\) possible auxiliary victims, which can be resolved with
at most \(m/2-2\) additional candidate queries. Hence
\[
\widetilde Q_{\mathrm S}^{\star}(0)\le \frac m2.
\]

For the lower bound, consider a deterministic zero-error tree of depth
\(q<m/2\) and restrict it to the core victims. Before the first
\(p_i\) query, follow the non-target response to every \(v_z\) query.
If \(k\) such queries occur, at least \(2m-2k\) core victims remain.
The tree cannot terminate on this path because every candidate is
successful for at most two core victims.

Let \(p_i\) be the first \(p\)-query. One of its two response cells
contains at least \(m-k\) of the remaining victims. The two full core
response cells of \(p_i\) are
\[
L_1(i)=\{c_i\}\cup\{c_{m+j}:j\ne i\},
\qquad
L_2(i)=\{c_{m+i}\}\cup\{c_j:j\ne i\}.
\]
Within either cell, every subsequent \(p\)-query is constant or
separates at most two victims from the rest, and every \(v\)-query has
a target cell of size at most two. Since every output is successful
for at most two core victims, a depth-\(r\) subtree can be correct on
at most \(2(r+1)\) victims, by induction on \(r\). With
\(r=q-k-1\), zero-error recovery therefore requires
$
m-k\le 2(q-k),
$
which implies \(q\ge m/2\), a contradiction. Thus
\[
\widetilde Q_{\mathrm S}^{\star}(0)=\frac m2.
\]

\paragraph{Binary feedback.}
Under binary feedback, every \(p_i\) query is successful exactly on
\(\{d_i,d_{(i+1)\bmod m}\}\), and every \(v_z\) query is successful
exactly on \(\{c_z,c_{(z+1)\bmod N}\}\). Along a sequence of \(q\)
non-target responses, at least \(H-2q\) victims remain possible.
Since no candidate is successful for more than two victims,
zero-error recovery requires $H-2q\le2,$
and hence \(q\ge H/2-1\).

For the matching upper bound, the candidates
$
\{p_0,p_2,\ldots,p_{m-2}\}
\cup
\{v_0,v_2,\ldots,v_{N-2}\}
$
have disjoint success sets that partition the \(H\) victims into
\(H/2\) pairs. Query all but one of these candidates and return the
remaining candidate after all non-target responses. Therefore
\[
\widetilde Q_{\mathrm B}^{\star}(0)=\frac H2-1.
\]


\end{proof}

\subsection{Proof of Corollary~\ref{cor:nonzero_delta}}
\label{app:nonzero_delta}

\begin{proof}
Suppose a randomized strategy satisfies Definition~\ref{def:budget} at an integer budget $q$. Stop it after $q$ queries, retaining its output if it has returned, and assigning a fixed candidate otherwise. Let $\widetilde Z(h,R)$ denote this output, where $R$ contains all internal randomness. For each $h$,
\[
 \{T_{\mathcal A}(h,R)\le q,\ Z_{\mathcal A}(h,R)=\phi(h)\}
 \subseteq\{\widetilde Z(h,R)=\phi(h)\}.
\]
Fixing $R=r$ gives a deterministic tree of depth at most $q$. For any nonempty subfamily $S$, define its number of successful states as
\[
 K_S(r)=\sum_{h\in S}\ind\{\widetilde Z(h,r)=\phi(h)\}.
\]
By Definition~\ref{def:budget} and the inclusion above, for every
\(h\in\mathcal H\),
\[
\Pr_R[\widetilde Z(h,R)=\phi(h)]\ge 1-\delta.
\]
Therefore, for any nonempty \(S\subseteq\mathcal H\),
\begin{equation}
(1-\delta)|S|
\le
\sum_{h\in S}\Pr_R[\widetilde Z(h,R)=\phi(h)]
=
\mathbb E_R[K_S(R)].
\label{eq:counting-average}
\end{equation}
It is enough to bound \(K_S(r)\) for every fixed \(r\) since
Equation~\ref{eq:counting-average} averages these bounds over the
randomness of the algorithm.

\paragraph{FAST.}
Take $S=C$. Each query has at most two responses on $C$, so the tree has at most $2^q$ nonempty leaves there. Each leaf returns a single trigger candidate and distinct core victims
have distinct unique functional trigger candidates. Hence each leaf can
be correct for at most one core victim. Consequently $K_C(r)\le\min\{N,2^q\}$ for every $r$, and Equation~\ref{eq:counting-average} yields
\[
 (1-\delta)N\le2^q,\qquad
 q\ge\max\{0,\lceil\log_2((1-\delta)N)\rceil\}.
\]

\paragraph{SLOW.}
Take $S=A$ as in Appendix~\ref{app:exact_separation}. Every query is constant or splits off one state. A restricted depth-$q$ tree of this form has at most $q+1$ nonempty leaves. At depth zero the bound is one; a splitting root contributes at most one singleton leaf and at most $q$ leaves below its other child by induction. Constant roots cannot increase this bound. Thus
\[
 K_A(r)\le\min\{m,q+1\},\qquad
 (1-\delta)m\le q+1,
\]
which gives $q\ge\max\{0,\lceil(1-\delta)m\rceil-1\}$. The FAST and SLOW upper bounds are their zero-error strategies. The SLOW lower bound uses only $A$ and need not be tight for the full family at $\delta>0$.

\paragraph{Binary feedback.}
Every query is an equality test on $\HH$, so the same argument gives
\[
 q\ge q_\delta:=\max\{0,\lceil(1-\delta)H\rceil-1\}.
\]
Since $0\le q_\delta\le H-1$, sample $q_\delta+1$ distinct candidates uniformly without replacement and order them uniformly. Query the first $q_\delta$, returning a target hit immediately; after all misses return the last candidate. For every $h$, this strategy uses at most $q_\delta$ queries and succeeds exactly when $\phi(h)$ is in the sampled list. Hence
\[
 \Pr_h[T\le q_\delta,\ Z=\phi(h)]
 =\frac{q_\delta+1}{H}\ge1-\delta.
\]
This also covers $q_\delta=0$, when the strategy returns a uniformly sampled candidate without querying.

For fixed $\delta<1$, $\log_2((1-\delta)N)=n+\log_2(1-\delta)$ and $(1-\delta)m$ is proportional to $m$. Combining the lower and upper bounds gives the stated asymptotic orders. These conclusions do not assert a uniform separation when $\delta$ approaches one as $H$ increases.
\end{proof}

\subsection{Proof of Proposition~\ref{prop:transfer}}
\label{app:neural-transfer}

\begin{proof}
Fix \(a\in\{\mathrm F,\mathrm S\}\). Let
\(\widehat{\mathcal I}_a\) denote the access model defined by the
trained classifiers and let \(\mathcal I_a\) denote the corresponding
finite access model. By construction, the two access models have the
same set of hidden victims, allowed queries, output candidates, and
query cost.

Consider any recovery algorithm \(\mathcal A\). Fix a hidden victim
\(h\) and a realization of the randomness of \(\mathcal A\). We show
by induction on the number of queries that \(\mathcal A\) produces the
same transcript under \(\widehat{\mathcal I}_a\) and
\(\mathcal I_a\).
Before the first query, the two transcripts are identical. Suppose
they are identical after \(k\) queries. Since \(\mathcal A\) has the
same transcript and the same fixed randomness in both executions, it
chooses the same next query \(q\). By
Equation~\eqref{eq:transfer_responses},
\[
\widehat O_h^a(q)=O_h^a(q).
\]
The next response is therefore the same in both executions. Hence the
transcripts remain identical after \(k+1\) queries.

It follows that the complete transcripts are identical. In particular,
\(\mathcal A\) has the same stopping time and returns the same output
in both access models. Equation~\eqref{eq:transfer_functional} gives
\[
\widehat F_\tau^a(h)=F(h)
\]
for every hidden victim \(h\). Therefore, for every
\(q\in\mathbb N_0\),
\[
\Pr_h\!\left[
T_{\mathcal A}^{\widehat{\mathcal I}_a}\le q,\,
Z_{\mathcal A}^{\widehat{\mathcal I}_a}
\in\widehat F_\tau^a(h)
\right]
=
\Pr_h\!\left[
T_{\mathcal A}^{\mathcal I_a}\le q,\,
Z_{\mathcal A}^{\mathcal I_a}\in F(h)
\right].
\]
Thus
\[
Q_{\mathcal A}^{\widehat{\mathcal I}_a}
(\widehat F_\tau^a,\delta)
=
Q_{\mathcal A}^{\mathcal I_a}
(F,\delta)
\]
for every \(\delta\in[0,1)\). Since the two access models allow the
same recovery algorithms, taking the infimum over \(\mathcal A\)
gives
\[
Q_{\widehat{\mathcal I}_a}^{\star}
(\widehat F_\tau^a,\delta)
=
Q_{\mathcal I_a}^{\star}
(F,\delta).
\]

Under binary feedback, Equation~\eqref{eq:transfer_responses} also
implies
\[
\ind\{\widehat O_h^a(q)=0\}
=
\ind\{O_h^a(q)=0\}
\]
for every \(h\) and \(q\). The same argument therefore gives equality
of the optimal recovery budgets under binary feedback.
\end{proof}




\section{Simulation Between Feedback Models}
\label{app:feedback_simulation}

Consider two access models \(\mathcal I_1\) and \(\mathcal I_2\).
They have the same allowed queries, allowed output candidates, success sets
\(F(h)\), and cost per query. Assume that the response kernels are memoryless:
conditional on the hidden victim \(h\) and the current query \(q\), the next
response under \(\mathcal I_i\) has distribution \(P_i(y\mid h,q)\),
independently of the previous transcript. The query may still be chosen
adaptively. We say that \(\mathcal I_1\) can simulate \(\mathcal I_2\)
if there exists a known stochastic response transformation
\(R(\widetilde y\mid y,q)\), independent of \(h\), such that
\[
P_2(\widetilde y\mid h,q)
=\sum_y R(\widetilde y\mid y,q)P_1(y\mid h,q)
\]
for every \(h,q,\widetilde y\).

\begin{proposition}[\textbf{Monotonicity under feedback simulation}]
\label{prop:feedback_simulation}
If \(\mathcal I_1\) can simulate \(\mathcal I_2\), then for every
\(\delta\in[0,1)\),
\[
Q^\star_{\mathcal I_1}(F,\delta)
\le
Q^\star_{\mathcal I_2}(F,\delta).
\]
\end{proposition}

\begin{proof}
Fix any recovery algorithm \(\mathcal A_2\) for \(\mathcal I_2\).
An algorithm \(\mathcal A_1\) for \(\mathcal I_1\) runs a copy of
\(\mathcal A_2\), maintaining its simulated transcript.
When this copy chooses query \(q_k\), algorithm \(\mathcal A_1\)
makes the same query and receives \(Y_k\). By the memoryless assumption,
conditional on the simulated history and \(q_k\), this response has
distribution \(P_1(\cdot\mid h,q_k)\).
It then draws \(\widetilde Y_k\) from
\(R(\cdot\mid Y_k,q_k)\) and supplies it to the copy of \(\mathcal A_2\).
The simulation condition makes the conditional distribution of
\(\widetilde Y_k\) equal to \(P_2(\cdot\mid h,q_k)\).
Induction therefore gives the same simulated transcript distribution
as an actual execution of \(\mathcal A_2\) under \(\mathcal I_2\),
for every hidden victim \(h\).
Using the same stopping and output rules and one query per simulated query gives
\[
\Pr_h[T_{\mathcal A_1}\le q,\ Z_{\mathcal A_1}\in F(h)]
=\Pr_h[T_{\mathcal A_2}\le q,\ Z_{\mathcal A_2}\in F(h)]
\]
for every \(h\) and \(q\in\mathbb N_0\). Hence
\(Q_{\mathcal A_1}(F,\delta)=Q_{\mathcal A_2}(F,\delta)\).
Taking the infimum over \(\mathcal A_2\) proves the claim.
\end{proof}

For label feedback and binary feedback, define $
\rho(y)=\mathbf 1\{y=t\}.
$
Binary feedback is obtained from label feedback by the deterministic
response transformation
$
R(b\mid y)
=
\mathbf 1\{b=\rho(y)\}.
$
Hence label feedback can simulate binary feedback. Proposition
\ref{prop:feedback_simulation} gives
$
Q^\star_{\mathrm{label}}(F,\delta)
\le
Q^\star_{\mathrm{binary}}(F,\delta).
$

Note that the same argument does not apply to FAST and SLOW. Their response
marginals satisfy
$
\Pr_{\mathrm{FAST}}[Y_q=y]
=
\Pr_{\mathrm{SLOW}}[Y_q=y]
$
under a uniformly chosen hidden victim for every fixed query \(q\).
This equality only matches the distribution before conditioning on the
hidden victim and adaptive transcript. It does not provide a
victim-independent response transformation \(R\) satisfying the
simulation condition above. Matched response marginals therefore do
not imply equal optimal recovery budgets.

\section{CIFAR-10 ResNet-18 Training and Audit Protocol}
\label{app:neural_protocol}

\subsection{Data, Model, and Trigger Codebooks}
\label{app:neural_setup}

We use the standard CIFAR-10 training and test sets with class \(0\) as
the target. Training images are augmented by random \(32\times32\)
crops with four-pixel padding and random horizontal flips. The trigger
is applied after augmentation and before normalization. Inputs are
normalized using channel means $(0.4914,\,0.4822,\,0.4465)
$ and standard deviations
$(0.2470,\,0.2435,\,0.2616).
$
The classifier is ResNet-18 with ten output classes. We replace the
initial ImageNet convolution by a \(3\times3\) convolution with stride
one and padding one. The initial max-pooling layer is removed.
Each trigger candidate is a binary \(4\times4\) pattern of Hamming
weight eight. The pattern overwrites the bottom-right \(4\times4\)
region of the image and is replicated across the three color channels.
We use fixed codebooks of size \(H=12,24,48\). Their minimum pairwise
Hamming distances are \(8\), \(6\), and \(6\), respectively.
A common clean base model is trained for \(100\) epochs with batch size
\(128\). Training uses SGD with Nesterov momentum \(0.9\), initial
learning rate \(0.1\), and weight decay \(5\times10^{-4}\). The
learning rate follows a cosine schedule. The same clean base is used
for all FAST and SLOW victims at all three values of \(H\).

\subsection{FAST and SLOW Training}
\label{app:neural_training}

For each hidden victim \(h\), the FAST and SLOW models start from the
same clean base and use the same random stream during fine-tuning.
They therefore receive the same augmented clean examples and the same
sequence of trigger candidates. The label assigned to a triggered
example is the controlled difference between the two models.
At each training step, a minibatch contains \(128\) clean examples.
Each clean example is copied twice to form \(256\) triggered examples.
Trigger candidates are sampled uniformly from the \(H\) candidates.
Clean examples retain their original CIFAR-10 labels. A triggered
example using candidate \(q\) receives label
$O_h^{\mathrm F}(q)
$
for the FAST model and
$
O_h^{\mathrm S}(q)
$
for the SLOW model.

Fine-tuning uses SGD with Nesterov momentum \(0.9\), learning rate
\(0.01\), and weight decay \(5\times10^{-4}\). The learning rate
follows a cosine schedule over the full fine-tuning run. We use
$s_{12}=750,
s_{24}=1500,
s_{48}=3000
$
steps for \(H=12,24,48\). Since each step contains \(256\) triggered
examples, the expected number assigned to any fixed candidate is
$
\frac{256s_H}{H}=16000
$
at every size. Hence increasing \(H\) does not reduce the expected
training exposure of each trigger candidate.

\subsection{Fixed Source and Functional Set Audits}
\label{app:neural_audit}

The fixed source \(x_0\) is selected using only the clean base model.
We choose the first CIFAR-10 test image whose ground-truth label is
not the target and that the clean base classifies correctly. This rule
selects test image \(0\) whose ground-truth label is \(3\). No FAST or
SLOW victim response is used in the selection.

For each trained victim \(h\), define the fixed source response
\[
\widehat O_h(q)=f_h(T_q(x_0)).
\]
We evaluate \(\widehat O_h(q)\) for every victim \(h\) and every
candidate \(q\), so each family contributes \(H^2\) audited
responses. For \(H=12,24,48\), this gives \(144\), \(576\), and
\(2304\) responses per family. Every audited response agrees with the
corresponding prescribed FAST or SLOW response table.

We evaluate functionality on the set
\[
X_{\mathrm{test}}
=
\{x:\text{the ground-truth CIFAR-10 label of }x\neq0\},
\]
which contains \(9000\) test images. For each victim \(h\) and trigger
candidate \(z\), define
\[
\widehat g_h(z)
=
\frac{1}{9000}
\sum_{x\in X_{\mathrm{test}}}
\mathbf 1\{f_h(T_z(x))=0\},
\]
and
\[
\widehat F_{0.9}(h)
=
\{z:\widehat g_h(z)\ge0.9\}.
\]

Membership in \(\widehat F_{0.9}(h)\) is determined exactly. Since the
source set contains \(9000\) images, a candidate is functional if and
only if it produces at least \(8100\) target predictions. Evaluation
of a candidate can therefore stop once either \(8100\) target
predictions have been observed or more than \(900\) non-target
predictions have been observed. The first event certifies membership
in \(\widehat F_{0.9}(h)\); the second certifies non-membership. The
intended candidate \(\phi(h)\) is evaluated on all \(9000\) images so
that its empirical ASR is also obtained exactly.

For every trained victim in both families and at all three values of
\(H\), the audit gives
\[
\widehat F_{0.9}(h)=\{\phi(h)\}.
\]
Thus the learned models reproduce both quantities required by
Proposition~\ref{prop:transfer}: the complete response table on the fixed source and the
functional set under the empirical success criterion. The optimal
zero-error recovery budgets of the finite FAST and SLOW problems
therefore transfer exactly to these learned fixed source recovery
problems.


\section{Additional Experiment}
\label{app:natural-policy}

This experiment compares binary feedback and label feedback for a
specified search policy on backdoored CIFAR-10 classifiers. It
examines the behavior of this policy and does not estimate the optimal
recovery budgets studied in the main text.

\subsection{Experimental Setup}

We use 24 matched sets of CIFAR-10 ResNet-18 classifiers with target
classes in \(\{0,3,7\}\). Each set contains models trained with
BadNet, GRASP, and a BELT-based cover variant, denoted BELT+
\citep{gu2017badnets,zhu2024gradient,qiu2024belt}. Each model is trained with
3,000 poisoned non-target examples whose labels are changed to the
target class. GRASP additionally uses 300 cover examples with their
original labels and Gaussian perturbations on the planted pixels.
BELT+ uses 3,000 cover examples with one of the nine planted pixels
removed while retaining their original labels.
For each trial, we use 32 source images to construct the response
features and 64 source images to evaluate candidate quality. The same
images are used under binary feedback and label feedback. The remaining
8,520 non-target CIFAR-10 test images are used only to evaluate the
candidate returned by the search.

A trigger candidate is a binary mask over an \(8\times8\) region.
Selected pixels are overwritten with white. Evaluating a candidate
requires predictions on the 32 images used for the response features
and the 64 images used for its search score, for a total of 96 queries. The search score is the fraction of the 64 predictions equal
to the target label.

Under label feedback, the predicted class labels on the 32 images are
retained. Under binary feedback, all non-target predictions are
represented by a single value. Therefore, the two feedback settings have exactly the same target predictions for the same candidate and source
images and differ only in the information retained about
non-target predictions.

\subsection{Search Policy}

The search policy is a LinUCB variant \citep{li2010contextual}. Let
\(p_t\) be the fraction of target predictions among the 32 source
images and let \(u\) contain the remaining 10 coordinates of the
response histogram, including the coordinate used for non-target
predictions under binary feedback. We define
\[
\widetilde u=
\begin{cases}
(1-p_t)u/\|u\|_2, & \|u\|_2>0,\\
0, & \|u\|_2=0,
\end{cases}
\]
and let \(r\in\mathbb R^{11}\) denote the vector consisting of
\(p_t\) and \(\widetilde u\). It satisfies
\[
\|r\|_2^2=p_t^2+(1-p_t)^2.
\]
Hence the response vectors have the same norm under the two feedback
settings for the same candidate and source images, while label
feedback can retain the distribution of the non-target labels.

For a candidate \(z\) and pixel coordinate \(j\), let \(d\) be the
fraction of selected pixels in \(z\), let
\(\sigma\in\{-1,+1\}\) indicate removal or addition of pixel \(j\),
and let \((\rho_j,\kappa_j)\) be its row and column coordinates scaled
to \([-1,1]\). With \(e_j\in\mathbb R^{64}\) denoting the coordinate
indicator, the feature vector is
\begin{equation}
x(z,j)=
(1,d,\sigma,\rho_j,\kappa_j,\,
e_j,\,
r,\,
\sigma r,\,
\rho_j r,\,
\kappa_j r)
\in\mathbb R^{113}.
\label{eq:natural-feature}
\end{equation}

The ridge model is initialized with \(A=I\) and \(b=0\). A possible
change with feature vector \(x\) receives the LinUCB score
\[
(A^{-1}b)^\top x
+
0.25\sqrt{x^\top A^{-1}x}.
\]
After the resulting candidate is evaluated, the reward
\(R\in\{-1,0,+1\}\) is the sign of the difference between its search
score and that of its parent candidate. The update is
\[
A\leftarrow A+xx^\top,
\qquad
b\leftarrow b+Rx.
\]

The search starts from 16 distinct random candidates whose numbers of
selected pixels are sampled from
$
\{1,2,4,8,12,16,24,32\}.
$
The two feedback settings use the same initial candidates. Among
evaluated candidates having an unvisited one-bit neighbor, the eight
with the highest search scores are retained. The unvisited one-bit
neighbors of these candidates are then scored by the ridge model and
the candidate with the largest score is evaluated next. A new random
candidate is also evaluated every 128 adaptive steps and whenever no
such neighbor remains.

At each query budget, the policy returns the evaluated candidate with
the highest search score. Its ASR is then measured on the 8,520
reporting images which are not used during the search.

\subsection{Results and Interpretation}
\label{app:natural-results}

Recovery succeeds when the candidate returned by the search policy has
ASR at least \(0.9\) on the 8,520 reporting images. We evaluate seven
query budgets from \(1{,}536\) to \(32{,}768\).
Table~\ref{tab:natural} reports three representative budgets and
Figure~\ref{fig:natural} shows the results at all seven budgets.

\begin{table}[t]
\centering
\caption{
Recovery success of the specified search policy under binary feedback
and label feedback. Success means that the returned candidate has ASR
at least \(0.9\). Each value averages 24 victim sets and five trials.
}
\label{tab:natural}
\begin{tabular}{c cc cc cc}
\toprule
& \multicolumn{2}{c}{BadNet}
& \multicolumn{2}{c}{GRASP}
& \multicolumn{2}{c}{BELT+} \\
\cmidrule(lr){2-3}
\cmidrule(lr){4-5}
\cmidrule(lr){6-7}
Query budget
& Binary & Label
& Binary & Label
& Binary & Label \\
\midrule
1,536  & 0.1083 & 0.1083 & 0      & 0      & 0 & 0 \\
8,192  & 0.3750 & 0.3750 & 0.0167 & 0.0167 & 0 & 0 \\
32,768 & 0.8667 & 0.8917 & 0.0833 & 0.0750 & 0 & 0 \\
\bottomrule
\end{tabular}
\end{table}

Under this policy, BadNet recovery becomes substantially more frequent
as the query budget increases under both feedback settings. The two
settings have the same success rate of \(0.3750\) at a budget of
\(8{,}192\). At the largest budget, the success rates are \(0.8667\)
under binary feedback and \(0.8917\) under label feedback.
Recovery is much less frequent for GRASP under the same policy. At the
largest budget, the success rates are \(0.0833\) under binary feedback
and \(0.0750\) under label feedback. For BELT+, no trial returns a
candidate with ASR at least \(0.9\) under either feedback setting at
any of the evaluated budgets.

\begin{figure}[t]
\centering
\includegraphics[width=.96\linewidth]{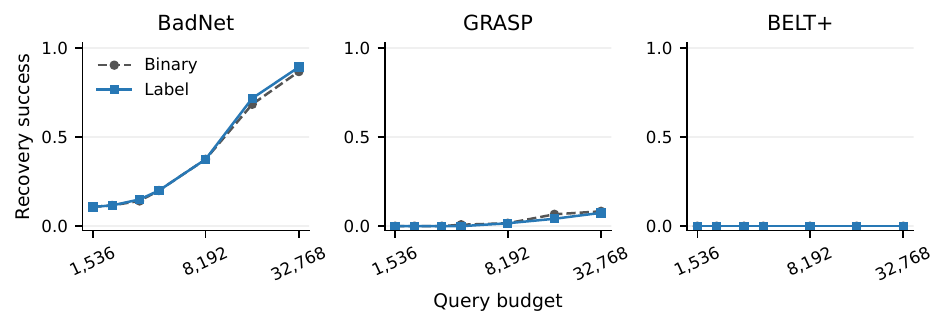}
\caption{
Recovery success at each query budget averaged over 24 victim sets
and five trials. Binary feedback and label feedback use the same
16 initial candidates. Later candidate choices can differ because
label feedback retains the non-target class labels used by the search
policy.
}
\label{fig:natural}
\end{figure}

Across the three backdoor methods, retaining the non-target class
labels does not produce a consistent change in recovery success for
this search policy. Label feedback gives slightly higher success for
BadNet at the largest budgets, slightly lower success for GRASP at the
largest budget, and the same zero success for BELT+.
This observation is consistent with
Proposition~\ref{prop:feedback_simulation}. An optimal algorithm with
label feedback can discard the additional class information and
simulate an algorithm using binary feedback, which implies
$
Q^\star_{\mathrm{label}}(F,\delta)
\le
Q^\star_{\mathrm{binary}}(F,\delta).
$
Note that the experiment above instead evaluates one specified search policy.
Although this policy uses the non-target class labels in its features,
the additional information need not improve its subsequent candidate
choices. The observed success rates describe this policy and
do not estimate the optimal recovery budget \(Q^\star\).

The same victim models were used during development of the search
policy and checkpoint selection used measurements from the CIFAR-10
test set. Therefore, the results should be interpreted as descriptive
results for these models rather than as an independent evaluation.

\end{document}